\documentclass[12pt,reqno]{amsart}
\usepackage{subfiles}
\usepackage[utf8]{inputenc}
\usepackage{t1enc}
\usepackage[margin=3.5cm]{geometry}
\usepackage{hyperref}
\usepackage{amsfonts,amsthm,amssymb,amsmath,amscd}
\usepackage{mathtools}
\usepackage{braket}
\usepackage{comment}
\numberwithin{equation}{section}

\theoremstyle{plain}
\newtheorem{theorem}{Theorem}
\newtheorem{lemma}[theorem]{Lemma}

\newtheorem{proposition}[theorem]{Proposition}
\theoremstyle{definition}

\newcommand{\DeltaB}{\boldsymbol{\Delta}}

\newcommand{\norm}[1]{\left\lVert#1\right\rVert}

\DeclareMathOperator{\conv}{conv}

\DeclareMathOperator{\vol}{vol}
\DeclareMathOperator{\dist}{dist}
\DeclareMathOperator{\relint}{relint}

\DeclareMathOperator{\aff}{aff}

\title[Intrinsic volumes of the quantum state space and MUBs]{Intrinsic volumes of the quantum state space and mutually unbiased bases}

\author[Zs.\ Szil\'agyi]{Zsombor Szil\'agyi}
\address{
MTA-BME Lendület “Momentum” Quantum Information Theory Research Group and\\
Department of Analysis and Operations Research, Institute of Mathematics and\\
Doctoral School of Mathematics and Computer Science, Budapest University of Technology and Economics, Műegyetem rkp. 3., H-1111 Budapest, Hungary and\\
Institute for Natural Sciences and Basic Subjects, Bánki Donát Faculty of Mechanical and Safety Engineering, Óbuda University, Bécsi út 96/B, 1034 Budapest, Hungary
}
\email{zsombor.szilagyi@gmail.com}

\author[M.\ Weiner]{Mih\'aly Weiner}
\address{
Department of Analysis and Operations Research, Institute of Mathematics,
Budapest University of Technology and Economics M\H{u}egyetem rkp. 3--9 H-1111 
and MTA-BME Lend\"ulet ``Momentum'' Quantum Information Theory Research Group}
\email{mweiner@math.bme.hu}

\thanks{Sz.Zs. and M.W. are supported by the Ministry of Culture and Innovation and the National Research, Development and Innovation Office within the Quantum Information National Laboratory of Hungary (Grant No. 2022-2.1.1-NL-2022-00004). M.W. is also supported by the NRDI grant K132097.}

\subjclass[2010]{}


\begin{document}

\begin{abstract}
Previous studies on the geometrical properties of the state space of a finite-level quantum system have determined its volume and surface area. Building on this foundation, we derive explicit formulas for two additional intrinsic volume quantities.

The question of whether a complete set of mutually unbiased bases exists in dimension $d$ can be equivalently framed as whether a specific convex polytope can be inscribed within the state space of a $d$-level quantum system. One motivation for our work was the hypothesis that a smaller intrinsic volume of the state space compared to the corresponding intrinsic volume of the mentioned polytope could rule out such an inscription. While our computations of these two intrinsic volumes do not lead to this conclusion, they nonetheless provide fundamental insights into the geometric structure of quantum state spaces. In particular, we show that these quantities can be used to rule out the existence of some unit-vector ``configurations'' 
(though not the one formed by the bases vectors of a complete set of mutually unbiased bases). 
\end{abstract}

\maketitle

\section{Introduction}\label{sec:intro}

Finding collections of unit vectors whose pairwise inner products have prescribed moduli is a recurrent theme across mathematics, physics, and computer science. Often the same algebraic condition is studied under different names by different communities. In the real Euclidean space~$\mathbb{R}^{d}$, for instance, the search for $2d$ unit vectors ${\psi_{(r,j)}} (r\in\{1,2\},\, j\in [d]\equiv \{1,\ldots d\})$ satisfying
$$
|\langle \psi_{(r,j)},\psi_{(s,k)}\rangle|^2=\left\{\begin{matrix}
 0, & \textrm{ if } r = s\\
 \frac{1}{d}, & \textrm{otherwise},
\end{matrix}\right.
$$
for all $(r,j)\neq (s,k)$, is simply the problem of constructing a real $d\times d$ Hadamard matrix; the entries of the matrix are given by $\sqrt{d}\,\langle\psi_{(1,j)},\psi_{(2,k)}\rangle$. The celebrated Hadamard conjecture asks whether such a matrix exists for every $d$ divisible by~$4$, and remains open despite a century of progress in combinatorial design theory~\cite{Hadamard1893,Horadam2007}.

Quantum information theory motivates the complex analogue of the same question. Throughout this paper the $d$-dimensional Hilbert space of a finite-level quantum system will be identified with~$\mathbb{C}^{d}$, endowed with the usual Hermitian inner product. While unit vectors of this space play a special role in modeling a finite-level quantum system, direct physical meaning is attached not to a unit vector $\psi \in \mathbb{C}^d $ itself, but rather to the rank one orthogonal projector $P=|\psi\rangle\langle\psi|$ projecting to the {\it ray} (i.e.\ one-dimensional subspace) spanned by $\psi$. In particular, for unit vectors $\psi,\psi'$, the quantity $|\langle \psi, \psi'\rangle|^2 = {\rm Tr}(P P')$ 
(where $P=|\psi\rangle\langle\psi|$ and $P'=|\psi'\rangle\langle\psi'|$) has operational meaning, whereas the complex phase of $\langle\psi,\psi'\rangle$ does not.

Quantum information theory has thus brought renewed interest to the existence of vector systems with specified moduli of pairwise inner products, in particular in two famous cases. Though considerably effort is invested by the scientific community, none of them has been fully clarified to this day. The mentioned cases are:

\begin{itemize}
\item[i)] \textbf{Symmetric, informationally complete POVMs (SIC‑POVMs).} Although originally considered as a collection of positive semidefinite matrices, in order to give a so-called SIC‑POVM in dimension~$d$, one needs to find $d^{2}$ unit vectors $\psi_1,\ldots \psi_{d^2}\in \mathbb{C}^d$ satisfying
$|\langle \psi_j, \psi_k\rangle|^2=\frac{1}{d+1}$ for all $j\neq k$. Apart from applications in quantum information theory, -- see e.g.\ \cite{fuchsSIC} and the references therein -- they have even appeared in topics like  compressed sensing in radar \cite{hermanStrohmerRadar}. Identifying whether SIC‑POVMs exist in all dimensions is a long‑standing open problem dating back to Zauner’s 1999 thesis~\cite{Zauner1999}. Exact constructions are now known up to $d=53$ and sporadically beyond, while high‑precision numerical solutions suggest existence at least up to $d=193$ \cite{Appleby2018,Grassl2020}; see more in the introduction of the recent work \cite{ApplebyFlammiaKopp2025}. The prevailing conjecture is that SIC‑POVMs exist for every $d$.

\item[ii)] {\bf Complete system of mutually unbiased bases (compl.\ MUB syst.). } For a complete $d$-dimensional MUB system one needs $d(d+1)$ unit vectors $\psi_{r,j}\in \mathbb{C}^d\, (r=1,\ldots d+1,\, j=1,\ldots d)$ forming $d+1$ mutually unbiased orthonormal bases; i.e.\ such that 
$$
|\langle \psi_{(r,j)},\psi_{(s,k)}\rangle|^2=\left\{\begin{matrix}
 0, & \textrm{ if } r = s\\
 \frac{1}{d}, & \textrm{otherwise},
\end{matrix}\right.
$$
for all $(r,j)\neq (s,k)$. Easy arguments show that more than $d+1$ MUBs cannot exists and in fact the number of bases in a collection of MUBs is $d+1$ if and only if the rank one projectors corresponding to the bases vectors 
span the set $M_d(\mathbb{C})$ of $d\times d$ matrices (which is why such a system is referred to as a {\it complete} one).

MUBs arise naturally in several quantum information protocols (e.g.\ in dense coding, teleportation, entanglement swapping, covariant cloning, and quantum
state tomography) and have been the subject of extensive investigation both from purely mathematical and quantum informational perspectives; see e.g.\ the review \cite{DurtEnglertBengtssonZyczkowski}. While exhibiting a triplet of MUBs is easy in every dimension $d\geq 2$, complete systems are only constructed in prime-power dimensions \cite{Ivanovic,WoottersFields}, with the maximal number of MUBs being unknown in all other dimensions. Numerical evidence and the analogies to finite affine planes \cite{2019rigidity} suggests that they only exist in prime power dimensions.
\end{itemize}

Of course, for a question of this type to be ``sensible'', the prescribed moduli of the pairwise inner products cannot be arbitrary. 
Some constraints follow directly from the fact that the corresponding projectors are elements of the set 
$$
S_1(\mathbb{C}^d)=\{A\in M_d(\mathbb{C}^d)| A=A^*,{\rm Tr}(A)=1\},
$$
which is the affine space of self-adjoint, trace-one matrices. The pairwise inner products between these projectors are fixed and non-negative, given by the squared moduli of the inner products between the original vectors. (Here we consider $S_1(\mathbb{C}^d)$ as an Euclidean space with the Hilbert-Schmidt inner product.) We shall refer to these constraints as the \emph{trivial requirements} (defined in Section~2). When these are satisfied, the original prescription defines a well-determined polytope $\mathcal P$ -- that is, a geometrical structure with fixed pairwise distances between its vertices — in a $(d^2 - 1)$‑dimensional real Euclidean space -- and a distinguished point $q\in \mathcal P$. The desired vector configuration then exists if and only if it is possible to inscribe $\mathcal P$ into the convex set of density matrices
$$
\mathcal{S}_d \equiv S^+_1(\mathbb{C}^d):=\{\rho\in S_1(\mathbb{C}^d)\,|\, \rho\geq 0\}
$$
such that  $q\in \mathcal P$ is mapped to $\frac{1}{d}I\in \mathcal{S}_d$.

In some cases — including both of the aforementioned ones — this condition on $q$ is automatically satisfied: for any inscription of $\mathcal{P}$ into $\mathcal{S}_d$, the distinguished point $q$ is necessarily mapped to $\frac{1}{d}I$. In particular, for SIC‑POVMs, the associated polytope is a regular simplex with $d^2$ vertices and $q$ as its center. The radius of the circumscribed sphere of this simplex matches that of $\mathcal{S}_d$, forcing any inscription to map $q$ to $\frac{1}{d}I$.

The same situation arises in the case of the polytope~$\mathcal{P}_d$ corresponding to a $d$-dimensional complete MUB system. Here again, $q \in \mathcal{P}_d$ is the center of the circumscribed sphere of the polytope, whose radius matches that of the circumscribed sphere of $\mathcal{S}_d$. Hence, any realization of~$\mathcal{P}_d$ as a configuration of projectors must take $q$ to $\frac{1}{d}I$, and the existence of a $d$-dimensional complete MUB system is equivalent to the possibility of inscribing~$\mathcal{P}_d$ into $\mathcal{S}_d$. This insightful geometric reformulation was first recognized and articulated by Bengtsson and Ericsson in~\cite{Bengtsson2005}, who also introduced the evocative name ``{\it complementarity polytope}'' for $\mathcal{P}_d$. Their work has played a foundational role in reframing the MUB existence problem within a convex-geometric context.

Thus, proving the non-existence of certain vector configurations -- and in particular, that of a complete MUB system -- boils down to showing that a certain polytope cannot be inscribed into the convex set of density matrices. 

It is therefore natural to look for quantities that are monotone under inscription. For example, a body of larger volume cannot be inscribed into one of a smaller volume. Somewhat less known, but for convex bodies, surface area also behaves in a monotonous manner. Actually, for convex bodies there are certain further natural lower dimensional analogues of these mentioned quantities, which, again, are monotone under inclusion.

Let $C\subset \mathbb R^D$ be a (non-empty) compact convex set and consider its {$\varepsilon$-neighborhood} $C_\varepsilon$; i.e.\! the set of points of $\mathbb R^D$ whose distance from $C$ is less than or equal to $\varepsilon$. It is well-known that the ($D$-dimensional) volume of $C_\varepsilon$ is a polynomial of $\varepsilon$ of order $D$:
$$
\vol_D(C_\varepsilon)= a_0(C) + a_1(C)\varepsilon + a_2(C)\varepsilon^2+\ldots +a_d(C) \varepsilon^d.
$$
Clearly, $a_0(C)=\vol_D(C_0)=\vol_D(C)$ is simply the volume of $C$, and it is also not difficult to see that $a_1(C)$ is actually the surface area of $C$. It turns out that the further coefficients are some sort of measures of the ``lower-dimensional contents'' of the body. The $N$-dimensional \emph{intrinsic volume} $V_N(C)$ is -- up to a scaling factor -- the coefficient $a_{D-N}(C)$:
$$
a_{D-N}(C) =  \chi_{_{D-N}}\, V_N(C).
$$
The factor $\chi_{_{D-N}}$ is independent of the body $C$; it is the (usual) volume of the $D-N$ dimensional unit ball (e.g.\ $\chi_{_2}=\pi$ and $\chi_{_3}=\frac{4}{3}\pi$) with the convention $\chi_{_0}=1$. This is included to make this quantity truly ``intrinsic'' (so that it would remain unchanged even if $C$ is isometrically embedded into a larger dimensional Euclidean space, and the epsilon-expansion of the volume of $C_\varepsilon$ is considered there). Since in our computations the embedding space is usually fixed to be the affine subspace generated by $C$, we will often just work directly with the coefficients 
$$
\widetilde{V}_N(C)\equiv a_{D-N}(C)=\chi_{_{D-N}}\, V_N(C),
$$
which we shall call the \emph{unnormalized} intrinsic volumes of $C$. Note that the ``conversion rate'' between $V_N(C)$ and $\widetilde{V}_N(C)$ depends not only on $N$, but on the dimension of the convex body $C$ (though not on the particular ``shape'' or ``size'' of $C$).

For a longer introduction to intrinsic volumes see e.g.\ \cite{HugWeil2020}. Here we only recall a few key properties:
\begin{itemize}
\item  {\it Monotonicity}: if $C_1\subset C_2$ then $V_N(C_1)\le V_N(C_2)$ for every $N$.
\item {\it Difficulty of computation:} apart from Euclidean balls and hypercubes, explicit formulas for all intrinsic volumes are known only for a few (high-dimensional) bodies such as the regular simplex and the so-called cross-polytope \cite{kabluchko2017expected};
\item {\it Multiple representations:} when $C$ is a polytope, intrinsic volumes can be expressed in terms of face-volumes and certain solid-angle quantities.
\end{itemize}
Note however, that explicit formula for the solid-angle given by the intersection of $k$ half-spaces is known only for $k\leq 3$; in fact, it might be that for the intersection of more half-spaces, in some sense, there exists no such formula at all (similar to how we have a formula for the solution of the quadratic, cubic and quartic equations, but not for the quintic). 
Consequently, we only have a more or less explicit algorithm for computing the ``true'' volume, the surface area, and the next two intrinsic volumes of a polytope.

The volume and the surface area of the quantum state space $\mathcal{S}_d$ (with respect to the Euclidean distance induced by the Hilbert Schmidt inner product), that is, $\vol_{D}(\mathcal{S}_d)=V_D(\mathcal{S}_d)$ and 
$$
\vol_{D-1}(\partial \mathcal S_d)=
\frac{d}{d\varepsilon}\vol_{D}(\mathcal{S}_{d,\varepsilon})|_{\varepsilon=0}=
2 V_{D-1}(\mathcal{S}_d),$$ where  
$D=d^2-1$ is the dimension of the convex body $\mathcal S_d$,
were computed by {\.Z}yczkowski, Sommers and Bengtsson
\cite{Zyczkowski2003, bengtsson2017geometry}:
\begin{align}
\label{eq:volume_of_Sd}
\vol_{D} (\mathcal S_d)= \sqrt{d}\, (2\pi)^{\frac{d(d-1)}{2}} \,\frac{\Gamma (1)\cdots \Gamma (d)}{\Gamma (d^2)}    \\ 
\label{eq:surface_of_Sd}
\vol_{D-1}(\partial \mathcal{S}_d)=
\sqrt{d-1}\, (2\pi)^{\frac{d(d-1)}{2}} \,\frac{\Gamma (1)\cdots \Gamma (d+1)}{\Gamma (d) \Gamma (d^2-1)}.
\end{align}

Their method factors out the symmetries of $\mathcal S_d$ implemented by the unitary group $U(d)$; this is analogous to using polar coordinates for a body of revolution, which reduces volume computations to lower-dimensional integrals. Similarly, the volume of the $\varepsilon$-neighborhood of the $D$-dimensional body $\mathcal S_d$ can be reduced to integrating a certain density on the $\varepsilon$-neighborhood of a regular simplex in a much lower $(d-1)$-dimensional space. 

Although the full integral is tractable for $d\le 3$, its domain is complicated and an explicit formula for general $d$ seems infeasible. What we realized however, is that one does not need an explicit formula for the polynomial $p_d(\varepsilon)=\vol_D(\mathcal S_d)$ for the computation of $p_d^{(n)}(0)$
for $n\leq 3$: these derivatives can be expressed as integrals on a much simpler domain, namely on a simplex. Using this idea we extended the work of {\.Z}yczkowski, Sommers and Bengtsson and obtained the following:
\begin{theorem}\label{theorem:main}
With $D=d^2-1$, $p_d(\varepsilon)= \vol_{D}(\mathcal{S}_{d,\varepsilon})$ and $p_d^{(n)}$ denoting the $n$-th derivative of $p_d$, we have 
\begin{align}
&p_d^{(2)}(0) = (d-1)d^{\frac{3}{2}} (2\pi)^{\frac{d(d-1)}{2}} \,\frac{\Gamma (1)\cdots \Gamma (d)}{\Gamma (d^2-2)}, \label{eq:main1} \\[10pt] 
&p_d^{(3)}(0) = d (d - 1)^{\frac{3}{2}}  (2\pi)^{\frac{d(d-1)}{2}}  \frac{\Gamma (1)\cdots \Gamma (d+1) }{\Gamma (d)\Gamma (d^2-3)}.  \label{eq:main2}
\end{align}
\end{theorem}
The above theorem provides an explicit expression for the two intrinsic volumes of $\mathcal S_d$ following its ``usual'' volume and surface are, since
\begin{align*}
V_{D-2}(\mathcal S_d)&=\frac{1}{2!}\frac{1}{\chi_{_2}}
\left(\frac{d}{d\varepsilon}\right)^2\vol_{D}(\mathcal{S}_{d,\varepsilon})|_{\varepsilon=0} = \frac{1}{2\pi} p_d^{(2)}(0),\\
V_{D-3}(\mathcal S_d)&=\frac{1}{3!}\frac{1}{\chi_{_3}}
\left(\frac{d}{d\varepsilon}\right)^3\vol_{D}(\mathcal{S}_{d,\varepsilon})|_{\varepsilon=0} = \frac{1}{8\pi}p_d^{(3)}(0).
\end{align*}
The reason we chose to present these formulas in the particular form as they appear in our theorem -- i.e.\ to give the value of the derivatives at zero, rather than the intrinsic volumes, which are the normalized coefficients of the polynomial $p_d$ -- are twofold. In part, because this is the way the result naturally emerges from the computation. More importantly, this form makes it easier to compare these findings with the already known volume \ref{eq:volume_of_Sd} and surface \ref{eq:surface_of_Sd} formulas, and spot the evident pattern. 

Interestingly, this pattern does not continue in a simple way. As mentioned, for $d\le 3$ the full polynomial $p_d$ can be computed explicitly. Evaluating the integral in Mathematica yields, for $d=3$,
\begin{align*}
\vol_{8}(\mathcal S_{3,\varepsilon})=\frac{\sqrt{3}\pi^3}{5040}+
\left(\frac{\sqrt{2}\pi^3}{105} \right) \varepsilon + \left(\frac{\pi^3}{5 \sqrt{3}} \right)\varepsilon^2+ \left(\frac{2\sqrt{2} \pi^3}{5} \right) \varepsilon^3 +
\\
\left( \frac{\sqrt{3} \pi^3}{4} +\frac{\pi^4}{3} \right) \varepsilon^4  + \left(\frac{3 \pi^3}
{\sqrt{2}}\right) \varepsilon^5 + \left( \frac{3 \sqrt{3} \pi^3}{8} + \frac{\pi^4}{3} \right) \varepsilon^6 + \left(\frac{18\sqrt{2} \pi^3}{35}  \right)\varepsilon^7 + \left(\frac{\pi^4 }{24}\right) \varepsilon^8.
\end{align*}
The constant and linear terms recover the known volume and surface formulas; the $\varepsilon^2$ and $\varepsilon^3$ coefficients agree with Theorem~\ref{theorem:main}. However, the coefficient of $\varepsilon^4$ already involves $\pi^4$ in addition to $\pi^3$, indicating that computing higher derivatives by the same method may not be possible and might require different ideas.

We next consider the intrinsic volumes of the complementarity polytope $\mathcal P_d$. Bengtsson and Ericsson computed its volume and surface area\footnote{Note that the formula for the surface area (equation (19) in their paper) contains an evident typographical error: in the denominator, the correct factor is $(d^2 - 2)!$ rather than $(d^2 - 1)!$.} in \cite{Bengtsson2005}: 
\begin{align*}
\vol_{D}(\mathcal{P}_d)&=  \frac{\sqrt{d}^{d+1}}{(d^2-1)!},     \\[10pt]
\vol_{D-1}(\partial \mathcal{P}_d)&=   \frac{\sqrt{d}^{d+2}\sqrt{d^2-1}}{(d^2-2)!},
\end{align*}
with $D=d^2-1$ as before. We extended their computations to obtain the next two unnormalized intrinsic volumes. (Note that the usual surface area is an unnormalized intrinsic volume: $\vol_{D-1}(\partial \mathcal{P}_d)=\widetilde{V}_{D-1}(\mathcal{P}_d)=2V_D(\mathcal P_d)$.)
\begin{theorem}\label{th:P}
For the $D=(d^2-1)$-dimensional convex polytope $\mathcal{P}_d$, we have
\begin{align}
&\widetilde{V}_{D-2}(\mathcal{P}_d) = \frac{ \sqrt{2 d^2 -d -2}  (d^2 - 1) d^{d/2+1}}{4\,(d^2-3)!} \alpha \label{eq:P1}, \\[10pt] 
&\widetilde{V}_{D-3}(\mathcal{P}_d) =  \frac{2\sqrt{3d^2-2d-3}(d^2-1)(d-2)d^{d/2+1}}{9\,(d^2-4)!} \arctan \sqrt{\tan \frac{3\alpha}{4} \tan^3 \frac{\alpha}{4}} \ + \nonumber \\ 
 &+\frac{2\sqrt{d^2-d-1}(d^2-1)(d-1) d^{d/2+2}}{3\,(d^2-4)!} \arctan \sqrt{\tan(\frac{\alpha}{2}+\frac{\beta}{4}) \tan(\frac{\alpha}{2}-\frac{\beta}{4}) \tan^2(\frac{\beta}{4})},  \label{eq:P2}
\end{align}
where $\alpha=\arccos \big(1 - \frac{d}{d^2-1} \big)$ and $\beta=\arccos \big(1 - \frac{2d}{d^2-1} \big).$
\end{theorem}
Both the derivations and the final expressions are long and combinatorially delicate, so we sought an independent verification. We implemented a Monte-Carlo method to estimate volumes and in turn, to obtain intrinsic volumes numerically. The numerical estimates confirm our formulas to several decimal places, and the code may be useful to other researchers for estimating further intrinsic volumes. We shall not explain here this computer-based method in detail, but we do provide the (commented) python-code we used; see \cite{our_code}.

Sadly, a straightforward evaluation of these formulas for small $d$, together with a short asymptotic analysis for large $d$, shows that $V_N(\mathcal P_d)\leq V_N(\mathcal S_d)$ for all $d=2,3,\ldots$ and $N=D,D-1,D-2,D-3$ with $D=d^2-1$. Thus the first four intrinsic volumes do not rule out an inscription of $\mathcal P_d$ into $\mathcal S_d$.

Although knowing the first four intrinsic volumes of $\mathcal S_d$ does not settle the MUB problem, these quantities still provide valuable geometric information. To illustrate their use, we construct four hypothetical configurations of unit vectors in $\mathbb C^6$, specified only by the moduli of their pairwise inner products. Each configuration satisfies the \emph{trivial requirements}, yet none actually exists. The first is excluded already by comparing the volume of $\mathcal S_6$ with the volume of the polytope determined by the hypothetical configuration. The second passes the volume test but is ruled out by the surface area. The third passes both volume and surface tests, but is excluded by the third intrinsic volume. The fourth shows that sometimes one has to go as far as the fourth intrinsic volume to rule out an inscription.

The paper is organized as follows. After this introduction, Section~2 (Preliminaries) explains in detail the \emph{trivial requirements} and how -- when they hold -- the original problem reduces to an inscription problem. We also recall how intrinsic volumes of a polytope can be expressed using exterior solid angles and face volumes. Section~3 contains the computations leading to Theorem~\ref{theorem:main} (the third and fourth intrinsic volumes of $\mathcal S_d$). Section~4 treats the complementarity polytope and proves Theorem~\ref{th:P}. In Section~5 we present the four example configurations in $\mathbb C^6$ and show their non-existence. The Appendix contains two technical sections: one evaluates Selberg-type integrals used in Section~3, and the other fills in computational details from Section~4.

\section{Preliminaries}\label{sec:prelim}

\subsection{The trivial requirements}

Suppose $\psi_1,\ldots \psi_n \in \mathbb C^d$ are unit-vectors and let $$M_{j,k}=|\langle \psi_j,\psi_k\rangle|^2$$ for $j,k\in \{ 1,\ldots n \}$. Clearly, each entry of $M$ is a non-negative number and the diagonal entries are all equal to $1$. However, one can deduce much more about $M$. Indeed, consider the matrices
$$A_k = |\psi_k\rangle\langle\psi_k| -\frac{1}{d}I \;\;\;\; (k=1,\ldots n).$$
On one hand, by a straightforward computation, their pairwise Hilbert--Schmidt inner product is
$$
\langle A_j, A_k\rangle_{HS} \equiv {\rm Tr}(A_j^* A_k) = M_{j,k} -\frac{1}{d}.
$$
Thus,
$$G=M-\frac{1}{d}J,$$ 
where $J$ is the matrix whose entries are all equal to $1$ (of size  $n\times n$), is positive semidefinite, since it is a Gram-matrix. On the other hand, the matrices 
$A_1,\ldots A_k$ are self-adjoint and traceless; thus, we may view them as a collection of elements of the $(d^2-1)$-dimensional Euclidean space (i.e.\ real inner product space)
$$
S_0(\mathbb{C}^d)=\{A\in M_d(\mathbb{C}^d)| A=A^*,{\rm Tr}(A)=0\}.
$$
This can be summarized as follows.
\smallskip

\noindent
{\bf Trivial requirement:} there exit some vectors $v_1,\ldots v_n$ of a Euclidean space satisfying ($j,k\in \{1,\ldots n\}$)
\begin{itemize}
    \item[i)] $\|v_k\|^2=\langle v_k,v_k\rangle = 1-\frac{1}{d}$,
    \item[ii)] $\langle v_j,v_k\rangle \geq -\frac{1}{d}$,
    \item[iii)] ${\rm dim}(\rm Span\{v_1,\ldots v_n\})\leq d^2-1$,
\end{itemize}
such that $M_{j,k}=\langle v_j,v_k\rangle + \frac{1}{d}$ for all $j,k\in\{1,\ldots n\}$; i.e., such that $M=G+\frac{1}{d}J$ where $G$ is the Gram matrix of $v_1,\ldots v_n$ and $J$ is the matrix whose entries are all equal to $1$.

Note that the above is trivially equivalent to saying that $M-\frac{1}{d}J$ is positive semidefinite and for all $j,k\in \{1,\ldots n\}$:
\begin{itemize}
    \item[i)] $M_{k,k} =1$,
    \item[ii)] $M_{j,k} \geq 0$,
    \item[iii)] ${\rm rk}(M-\frac{1}{d}I)\leq d^2-1$. 
\end{itemize}
Note further, that many seemingly additional requirements in fact are encoded in the positive semidefinitness of $M-\frac{1}{d}J$. For example, we know that there cannot be more than $d$ pairwise orthogonal unit vectors in $\mathbb C^d$; so one may wonder, if this also follows from the above requirements. The answer is yes: since the sum of the entries of a positive semidefinite matrix is always nonnegative, we have that 
$$
\sum_{j,k} (M_{j,k}-\frac{1}{d})\geq 0 \;\; \Rightarrow \;\; \sum_{j,k}|\langle \psi_j,\psi_k\rangle|^2\geq \frac{n^2}{d},
$$
which is a -- by the way, well-known -- quantitative strengthening of the statement that a $d$-dimensional space can contain at most $d$ pairwise orthogonal unit vectors. In fact, a lot of results can be derived just by skillfully exploiting these conditions; e.g.\ this is how the main theorem of \cite{MatolcsiWeinerMUB21} is achieved. 

We shall know show that when these conditions are satisfied, the original question about the existence of a collection of unit vectors of $\mathbb C^d$ with prescribed moduli can be turned into a problem of inscription. This is well-known to experts of the field, but we include here a formal statement with a short proof in part, because of self-containment, and in part, to fix conventions and notations.
\begin{lemma}
Suppose $v_1,\ldots v_n$ is a collection of vectors of a Euclidean space satisfying the above listed trivial requirements, and let $\mathcal P$ be the polytope formed by the convex hull of
$\{0\}\cup \{v_1,\ldots v_n\}$, with distinguished point $q\equiv 0\in \mathcal P$. Then there exists a collection of unit vectors $\psi_1,\ldots \psi_n\in \mathbb C^d$ such that 
$$
|\langle \psi_j,\psi_k\rangle|^2 = \langle v_j,v_k\rangle + \frac{1}{d}\;\;\; (j,k\in\{1,\ldots n\})
$$
if and only if $\mathcal P$ can be (isometrically) inscribed into $\mathcal S_d$ in a manner that maps $q\in \mathcal P$ to $\frac{1}{d}I\in \mathcal S_d$.
\end{lemma}
\begin{proof}
The ``only if'' part has been already explained when we introduced the trivial requirements. What we still need to show is the ``if'' part. So suppose $\Phi$ is an inscription of $\mathcal P$ into $\mathcal S_d$ such that $\Phi(q)=\frac{1}{d}I$. Then for every $k\in\{0,\ldots n\}$, $\Phi(v_k)$ is a density matrix whose Hilbert-Schmidt distance from $\frac{1}{d}I$ is equal to the distance of $v_k$ from $q$, which is $\sqrt{1-\frac{1}{d}}$. Hence
\begin{eqnarray*}
1-\frac{1}{d}={\rm Tr}\big((\Phi(v_k)-\frac{1}{d}I)^2\big)=
{\rm Tr}\big(\Phi(v_k)^2-\frac{2}{d}\Phi(v_k)+\frac{1}{d^2}I\big)\\
= {\rm Tr}\big(\Phi(v_k)^2\big)-\frac{1}{d},   
\end{eqnarray*}
showing that $\Phi(v_k)$ is a density matrix whose square has trace equal to one; i.e.\ it is a rank one projection. Therefore, for every $k\in\{0,\ldots n\}$, there exists 
a unit vector $\psi_k\in \mathbb C^d$ such that $\Phi(v_k)=|\psi_k\rangle\langle\psi_k|$
Then 
\begin{eqnarray*}
2|\langle \psi_j,\psi_k\rangle|^2 &=& 2\langle \Phi(v_j),\Phi(v_k)\rangle_{HS} = 
2\langle\Phi(v_j)-\frac{1}{d}I,\Phi(v_k)-\frac{1}{d}I\rangle_{HS}+\frac{2}{d}
\\
&=&\|\Phi(v_j)-\Phi(q)\|^2_{HS}
+
\|\Phi(v_k)-\Phi(q)\|^2_{HS}-\|\Phi(v_k)-\Phi(v_j)\|^2_{HS}+\frac{2}{d}\\
&=&\|v_j\|^2+\|v_k\|^2-\|v_j-v_k\|^2+\frac{2}{d}=2 \langle v_j,v_k\rangle + \frac{2}{d},
\end{eqnarray*}
which concludes our proof.
\end{proof}

\subsection{Intrinsic volumes of polytopes}
In this subsection we shortly summarize and explain an alternative way to express (and to compute) the intrinsic volumes of polytopes. We will use the notation for the $\varepsilon$-neighborhood of a set $K\subset \mathbb{R}^D$
$$
K_{\varepsilon}=K +  \varepsilon B^D
$$
where $B^k$ denotes the $k$-dimensional unit ball, and
$$
K_{\varepsilon, \aff} = K_{\varepsilon} \cap \aff (K)
$$
when it is considered in the affine subspace spanned by $K$. In this section, we shall always assume that $K$ is a convex set, and that the affine subspace generated by $K$ is the full space $\mathbb R^D$; that is, $K$ is a $D$-dimensional convex body.

As mentioned, the volume of $K_{\varepsilon}$ is a polynomial in $\varepsilon$ of degree $D$:
$$
\vol_D(K_\varepsilon)= \sum_{k=0}^D \widetilde{V}_{D-k}(K) \varepsilon^k.
$$
The above expression is often referred to as the ``Steiner Formula''. The intrinsic volumes of $K$ are the normalized coeffitients $V_{D-k}(K)=\frac{1}{\chi_{_{k}}}\widetilde{V}_{D-k}(K)$ where $\chi_{_k}=\vol_k(B^k)$.

In case $K=P$ is a convex polytope, alternative formulas can be given for these coeffitients by decomposing the $\varepsilon$-neighborhood into a union of disjoint sets:
\begin{align*}
P_\varepsilon = \bigcup_{k=0}^{D} R_{k,\varepsilon}, \quad 
R_{k,\varepsilon} := \{x \in P_\varepsilon &\,\mid\, 
\text{closest point of $P$ to $x$ lies on the interior} \\[-10pt]
&\quad \text{of a $k$-dimensional face}\}.
\end{align*}
More explicitly 
\begin{equation}
R_{k,\varepsilon} = \bigcup_{\substack{F  \in \mathcal{F}_k(P)}} \left( \relint{F} + \left( N_P(F) \cap \varepsilon B^D \right) \right)  \label{eq:decomp}
\end{equation}
where $\mathcal{F}_k(P)$ is the set of $k$-dimensional faces of $P$ and $N_P(F)$ is the normal cone defined as the set of all outward normals to supporting hyperplanes of $P$ that contain the face $F$; that is,
$$
N_P(F) =  \{ u \in \mathbb{R}^D \mid \langle u, y - x \rangle \le 0 \quad \forall y \in P \} \subset F^\perp,
$$
where $x$ is an arbitrary point of the relative interior of $F$, and $F^\perp$ is the orthogonal complement of the linear subspace parallel to $F$. Since the face and the corresponding normal cone are orthogonal to each other, 
we have 
$$
\vol_{D}(R_{D-k,\varepsilon})=\sum_{F \in \mathcal{F}_{D-k}(P)} \vol_{D-k}(F) \cdot \vol_{k}(N_P(F)\cap \varepsilon B^D).
$$
Moreover, as $N_P(F)$ is a cone, it is unchanged under dilation, $N_P(F)\cap \varepsilon B^D= \varepsilon (N_P(F)\cap B^D)$, and hence
$\vol_{k}(N_P(F)\cap \varepsilon B^D)=\varepsilon^k\vol_{k}(N_P(F)\cap B^D)$. Putting everything together,
one concludes that
\begin{equation}\label{eq:VPoly}
V_{D-k}(P)=\sum_{F \in \mathcal{F}_{D-k}(P)} \vol_{D-k}(F) \frac{\vol_{k}(N_P(F)\cap B^D)}{\vol_{k}(B^k)}
\end{equation}
The fraction appearing in the above formula is the so-called {\it solid angle} of the normal cone $N_P(F)$.

\section{Statespace}\label{sec:statespace}

Any density matrix can be diagonalized by a unitary rotation, i.e., for all $\rho \in \mathcal{S}_d$ there exist a unitary matrix $U$ and a diagonal matrix 
$\Lambda$ such that
\[
\rho = U \Lambda U^{-1}.
\]
However, this decomposition is not unique. Indeed, for any diagonal unitary matrix $B$, we have $U \Lambda U^{-1} = (UB)\Lambda(UB)^{-1}$. What is unique is the multiset of diagonal elements of $\Lambda$; that is, the 
eigenvalues of $\rho$ with their multiplicities. 
This multiset can be viewed as an element of $\Delta^{d-1} / \!\sim$, where 
$\Delta^{d-1}$ denotes the standard $(d-1)$-dimensional simplex, and the 
equivalence relation $\sim$ identifies points in $\Delta^{d-1}$ that differ 
only by a permutation of their coordinates (i.e., different orderings of the 
same eigenvalues). 
Following the approach of \cite{Zyczkowski2003}, the state space 
$\mathcal{S}_d$, together with the Lebesgue measure $d\lambda^{d^2}$, can thus 
be identified with
$$
(\Delta^{d-1}/ \sim)  \times (\mathcal{U}(d)/\mathcal{U}(1)^d)
$$
with the measure
$$
 d\mu \times d\nu_{\text{Haar}} 
$$
where
$$
d\mu(x)= f_d(x) d\lambda^d
$$
and
\[
f_d(x)\equiv f_d(x_1,\ldots,x_d)=\prod_{1\leq i<j\leq d} (x_i-x_j)^2
\]
is the density function with respect to de Lebesgue measure. So the volume can be calculated as
$$
\vol_d(\mathcal{S}_d) = \frac{1}{d!} \vol_{d-1}(\Delta^{d-1})     \vol_{d(d-1)} (Fl_{\mathbb{C}}^d)  
$$
where $Fl_{\mathbb{C}}^d \equiv \mathcal{U}(d)/\mathcal{U}(1)^d$ called Flag manifold and
\[
\vol_{d(d-1)} (Fl_{\mathbb{C}}^d) =\frac{(2\pi)^{d(d-1)/2}}{1!2!\dots (d-1)!}
\]

\begin{proof}[Proof of Theorem~\ref{theorem:main}]
For any $A \in S_1(\mathbb{C}^d)$, the closest point in the state space 
$\widetilde A \in S_1^+(\mathbb{C}^d)$ is diagonal in the same basis. 
Indeed, since the Hilbert--Schmidt norm is unitarily invariant, 
minimizing $\|A - X\|$ over $X \in S_1^+(\mathbb{C}^d)$ is equivalent to 
minimizing $\|\operatorname{diag}(\lambda) - X'\|$, 
where $A = U \operatorname{diag}(\lambda) U^*$ and $X' = U^* X U$. 
Writing the squared Hilbert--Schmidt norm as
\[
\|\operatorname{diag}(\lambda) - X'\|^2
  = \sum_{i=1}^d (\lambda_i - X'_{ii})^2 
    + \sum_{i \neq j} |X'_{ij}|^2,
\]
we see that for any fixed diagonal entries $(X'_{ii})_{i=1}^d$, 
the off-diagonal term $\sum_{i \neq j} |X'_{ij}|^2$ is nonnegative and can only increase the distance. Hence, an optimal $X'$ has $X'_{ij} = 0$ for $i \neq j$ and is therefore diagonal in the same eigenbasis as $A$. 

As a consequence, the $\varepsilon$-neighborhood of the state space 
$\mathcal{S}_{d,\varepsilon}$ can be identified with
\[
(\Delta^{d-1}_{\varepsilon, \mathrm{aff}} / \!\sim) 
\times \bigl(\mathcal{U}(d) / \mathcal{U}(1)^d\bigr),
\]
where the equivalence relation $\sim$ is defined as above. (Here, the $\varepsilon$-neighborhood of $\mathcal{S}_d$ is understood as a subset of $S_1(\mathbb{C}^d)$.) Consequently, its volume can be calculated as
\[
p_d(\varepsilon)=\vol_D(\mathcal{S}_{d,\varepsilon}) 
  = \frac{1}{d!} 
    \vol_{d-1}(\Delta^{d-1}_{\varepsilon, \mathrm{aff}}) \,
    \vol_{d(d-1)}(Fl_{\mathbb{C}}^d).
\]
The second and third derivatives of $p_d$ at zero are
\begin{align}
p_d^{(2)}(0)  
  &= \frac{1}{d!} 
     \vol_{d(d-1)}(Fl_{\mathbb{C}}^d)
     \left. \frac{d^2}{d\varepsilon^2} 
     \vol_{d-1}(\Delta^{d-1}_{\varepsilon, \mathrm{aff}}) 
     \right|_{\varepsilon=0}, 
     \label{eq:V2} \\
p_d^{(3)}(0)
  &= \frac{1}{d!} 
     \vol_{d(d-1)}(Fl_{\mathbb{C}}^d)
     \left. \frac{d^3}{d\varepsilon^3} 
     \vol_{d-1}(\Delta^{d-1}_{\varepsilon, \mathrm{aff}}) 
     \right|_{\varepsilon=0}. 
     \label{eq:V3}
\end{align}
To evaluate these derivatives, it suffices to determine the 
$\varepsilon$-expansion of the integral up to degree four. 
We decompose the integral as
\begin{align*}
\vol_{d-1}(\Delta^{d-1}_{\varepsilon, \mathrm{aff}}) 
 &= \int_{\Delta^{d-1}_{\varepsilon, \mathrm{aff}}} f_d \, d\lambda^{d-1} 
    = \sum_{k=0}^{d-1} \int_{R_{k,\varepsilon}} f_d \, d\lambda^{d-1} \\
 & = \int_{R_{d-1,\varepsilon}} f_d \, d\lambda^{d-1} 
    + \int_{R_{d-2,\varepsilon}} f_d \, d\lambda^{d-1} 
    + \mathcal{O}(\varepsilon^4),
\end{align*}
by decomposing the domain of integration according to
formula~\eqref{eq:decomp} in Section~\ref{sec:prelim} as the disjoint union
\[
\Delta^{d-1}_{\varepsilon, \mathrm{aff}}
  = \bigcup_{k=0}^{d-1} R_{k,\varepsilon},
  \qquad
  R_{k,\varepsilon} 
  := \bigcup_{\substack{F \in \mathcal{F}_k(P)}} 
     \left( \relint F + \bigl( N_{\Delta^{d-1}}(F) \cap \varepsilon B^d \bigr) \right),
\]
and each $R_{k,\varepsilon}$ consists of points in the $\varepsilon$-neighborhood closest to a $k$-face of the simplex.

The term corresponding to $k = d-1$ in the above sum is constant, since 
$R_{d-1,\varepsilon} = \Delta^{d-1}$ is independent of $\varepsilon$; this term yields 
the volume of the state space itself. 
All terms in the sum~(ref) coming from faces of 
dimension less than $d-2$ contribute only to higher-order terms 
(of degree $\geq 4$) in the $\varepsilon$-expansion of the integral. 
Indeed, the function $f_d$ is a product of squared differences, so it 
vanishes to at least second order ($\varepsilon^2$) near any face of 
dimension less than $d-2$, since such faces lie in coordinate subspaces 
where some $x_i = x_j$. 
At the same time, the corresponding neighborhood around a $k$-face 
(with $k < d-2$) contributes a volume of order at least $\varepsilon^2$. 
Combining the vanishing behavior of $f_d$ with the dimensionality of the 
neighborhood yields a contribution of order at least $\varepsilon^4$ for 
these faces. 

Hence, the only term contributing to the second and third derivatives 
in~(\ref{eq:V2})--(\ref{eq:V3}) corresponds to the facets of the simplex. 
Due to the symmetry of the function $f_d$ and of the domain, the integral 
can be rewritten as $d$ times the integral over one specific facet 
$\Delta^{d-2}$—namely, the facet spanned by the first $d-1$ basis vectors. 
At the same time, we rewrite the integral using the \emph{coarea formula} (\cite{morganGMT, federerGMT}) by decomposing the domain as 
$R_{d-2,\varepsilon} = \Delta^{d-2} \times [0, \varepsilon]$. 
\[
I(\varepsilon) := \int_{R_{d-2,\varepsilon}} f_d \, d\lambda^{d-1} 
   = d \int_{\Delta^{d-2}} \! \int_0^\varepsilon 
     f_d(x + t v) \, dt \, d\lambda^{d-2},
\]
where $v = \frac{1}{\sqrt{d(d-1)}}(1, \dots, 1, 1-d)$ is the unit normal 
vector to the facet $\Delta^{d-2}$. 
Differentiating with respect to $\varepsilon$ gives
\begin{align}
I^{(1)}(\varepsilon)  
  &= d \int_{\Delta^{d-2}}  f_d(x + \varepsilon v) \, d\lambda^{d-2}, 
    \nonumber \\
I^{(2)}(0)
  &= d \int_{\Delta^{d-2}} 
     \left. \partial_\varepsilon f_d(x + \varepsilon v) \right|_{\varepsilon=0} 
     d\lambda^{d-2}, 
     \label{eq:deriv_1} \\
I^{(3)}(0)
  &= d \int_{\Delta^{d-2}} 
     \left. \partial_\varepsilon^2 f_d(x + \varepsilon v) \right|_{\varepsilon=0} 
     d\lambda^{d-2}.
     \label{eq:deriv_2}
\end{align}
The first and second directional derivatives appearing in these integrals 
can be expressed in terms of the gradient and the Hessian matrix of $f_d$:
\begin{align}
\left. \partial_\varepsilon f_d(x + \varepsilon v) \right|_{\varepsilon=0} 
  &= \langle \nabla f_d(x), v \rangle, \\
\left. \partial_\varepsilon^2 f_d(x + \varepsilon v) \right|_{\varepsilon=0} 
  &= \langle v, H_{f_d}(x) \, v \rangle.
\end{align}

We now derive the second derivative of $p_d$ at zero~\eqref{eq:main1} by evaluating \eqref{eq:deriv_1}. We decompose the normal vector as 
$v = s - e_d$, where 
$s = \frac{1}{\sqrt{d(d-1)}}(1, \dots, 1)$ and 
$e_d = \frac{1}{\sqrt{d(d-1)}}(0, \dots, 0, d)$. 
With this decomposition, the directional derivative simplifies as
\[
\langle \nabla f_d(x), v \rangle
  = \langle \nabla f_d(x), s \rangle 
    - \langle \nabla f_d(x), e_d \rangle
  = -\sqrt{\frac{d}{d-1}}\, \partial_d f_d(x),
\]
since the first term vanishes due to the symmetry of $f_d$. 
Using the recursive factorization
\[
f_d(x_1, \ldots, x_d)
  = f_{d-1}(x_1, \ldots, x_{d-1}) 
    \prod_{j=1}^{d-1} (x_j - x_d)^2,
\]
we obtain
\begin{align}\label{eq:firstderiv}
\partial_d f_d(x)
  &= f_{d-1}(x_1, \ldots, x_{d-1})\,
     \partial_d \!\left(
       \prod_{j=1}^{d-1} (x_j - x_d)^2
     \right) \\ \nonumber
  &= -2 f_{d-1}(x_1, \ldots, x_{d-1})
     \sum_{j=1}^{d-1} (x_j - x_d)
       \prod_{\substack{k=1 \\ k \neq j}}^{d-1} (x_k - x_d)^2. 
\end{align}
The integral in~\eqref{eq:deriv_1} is taken over the facet $\Delta^{d-2}$, where $x_d = 0$. Substituting into the previous expression yields
\[
2 \sqrt{\frac{d}{d-1}}\,
  f_{d-1}(x_1, \ldots, x_{d-1})
  \sum_{j=1}^{d-1} 
    x_j \prod_{\substack{k=1 \\ k \neq j}}^{d-1} x_k^2.
\]
By symmetry of both the integrand and the domain, each term in the sum 
contributes equally. Hence, we may take $(d-1)$ times the last term 
($j = d-1$), obtaining
\[
(d-1)\, x_{d-1} \prod_{k=1}^{d-2} x_k^2 
  = (d-1)\, x_1 x_2 \cdots x_{d-2} 
    \prod_{i=1}^{d-1} x_i.
\]
Substituting this into~\eqref{eq:deriv_1} and inserting the expression 
for $f_{d-1}$ gives
\[
2 \sqrt{d(d-1)} 
  \int_{\Delta^{d-2}}
    x_1 x_2 \cdots x_{d-2}
    \prod_{i=1}^{d-1} x_i
    \prod_{i<j}^{d-1} (x_i - x_j)^2 
    \, d\lambda^{d-2}.
\]
The integral above is of the same form as that appearing in 
Proposition~\ref{prop:main} in Appendix~\ref{appendix:selberg}, with corresponding parameters
\[
I_{\mathrm{Selberg}}(d-1, 2, 1, 0, d-2)
  = \frac{\Gamma(d+1)}{2\, \Gamma(d^2 - 2)} 
    \prod_{i=1}^{d} \Gamma(i)^2.
\]
Substituting this result into~\eqref{eq:V2} yields the desired formula for~\eqref{eq:main1}.

We now derive the third derivative of $p_d$ at zero~\eqref{eq:main2} by evaluating \eqref{eq:deriv_2}. 
We calculate the second directional derivative using the same decomposition 
of the normal vector $v = s - e_d$, where 
$s = \frac{1}{\sqrt{d(d-1)}}(1, \dots, 1)$ and 
$e_d = \frac{1}{\sqrt{d(d-1)}}(0, \dots, 0, d)$. 
Then
\[
\langle v, H(x) v \rangle
  = \langle s, H(x) s \rangle 
    - \langle e_d, H(x) s \rangle 
    - \langle s, H(x) e_d \rangle 
    + \langle e_d, H(x) e_d \rangle
  = \frac{d}{d-1}\, \partial_d^2 f(x),
\]
since, by symmetry of the Hessian $H(x)$, the first three terms vanish
(the vectors $H(x)s$ and $sH(x)$ are both zero).  
To compute $\partial_d^2 f(x)$, we differentiate~\eqref{eq:firstderiv} 
again. With a slight abuse of notation, we write $f_{d-1}$ for 
$f_{d-1}(x_1, \ldots, x_{d-1})$:
\begin{align*}
\partial_d^2 f(x)
  &= -2 f_{d-1} 
     \sum_{j=1}^{d-1}
       \partial_d \!\left[
         (x_j - x_d)
         \prod_{\substack{k=1 \\ k \neq j}}^{d-1} (x_k - x_d)^2
       \right] \\
  &= -2 f_{d-1}
     \sum_{j=1}^{d-1}
       \Big[
         -\!\prod_{\substack{k=1 \\ k \neq j}}^{d-1} (x_k - x_d)^2
         + (x_j - x_d)
           \partial_d \!\left(
             \prod_{\substack{k=1 \\ k \neq j}}^{d-1} (x_k - x_d)^2
           \right)
       \Big] \\
  &= -2 f_{d-1}
     \sum_{j=1}^{d-1}
       \Big[
         -\!\prod_{\substack{k=1 \\ k \neq j}}^{d-1} (x_k - x_d)^2
         + (x_j - x_d)
           \Big(
             -2 \sum_{\substack{k=1 \\ k \neq j}}^{d-1}
                 (x_k - x_d)
                 \prod_{\substack{l=1 \\ l \neq j,\, l \neq k}}^{d-1}
                   (x_l - x_d)^2
           \Big)
       \Big].
\end{align*}
Simplifying the signs, we obtain
\[
\partial_d^2 f(x)
  = 2 f_{d-1} 
    \left(
      \sum_{j=1}^{d-1} 
        \prod_{\substack{k=1 \\ k \neq j}}^{d-1} (x_k - x_d)^2
      + 2 \sum_{j=1}^{d-1} (x_j - x_d)
          \sum_{\substack{k=1 \\ k \neq j}}^{d-1}
            (x_k - x_d)
            \prod_{\substack{l=1 \\ l \neq j,\, l \neq k}}^{d-1}
              (x_l - x_d)^2
    \right).
\]
The integral in~\eqref{eq:deriv_2} is taken over the facet $\Delta^{d-2}$, where $x_d = 0$. Substituting $x_d = 0$ into the previous expression yields
\[
2 f_{d-1} \Bigg(
    \sum_{j=1}^{d-1} \prod_{\substack{k=1 \\ k \neq j}}^{d-1} x_k^2
  + 2 \sum_{j=1}^{d-1} x_j \sum_{\substack{k=1 \\ k \neq j}}^{d-1} x_k 
      \prod_{\substack{l=1 \\ l \neq j,\, l \neq k}}^{d-1} x_l^2
\Bigg).
\]
In the second sum, we can rewrite $\sum_{j=1}^{d-1} x_j \sum_{k \neq j} x_k 
= 2 \sum_{j<k}^{d-1} x_j x_k$, so the integrand in~\eqref{eq:deriv_2} becomes
\[
\frac{d}{d-1} \, 2 f_{d-1} \Bigg(
    \sum_{j=1}^{d-1} \prod_{\substack{k=1 \\ k \neq j}}^{d-1} x_k^2
  + 4 \sum_{j<k}^{d-1} x_j x_k \prod_{\substack{l=1 \\ l \neq j,\, l \neq k}}^{d-1} x_l^2
\Bigg).
\]
By symmetry of the integrand and the domain, each term in the first sum contributes equally; 
so it suffices to consider $(d-1)$ times the last term. 
Similarly, for the second sum, each term contributes equally, giving $\binom{d-1}{2}$ times the last term. 
Hence, the integrand can be replaced by
\begin{align*}
\frac{d}{d-1} \, 2 f_{d-1} 
  \Big(
    (d-1) \prod_{k=1}^{d-2} x_k^2
  + 4 \binom{d-1}{2} x_{d-2} x_{d-1} \prod_{l=1}^{d-3} x_l^2
  \Big)
= \\
=2d f_{d-1} \prod_{k=1}^{d-2} x_k^2 + 4 d(d-2) f_{d-1} x_{d-2} x_{d-1} \prod_{l=1}^{d-3} x_l^2.
\end{align*}
After substituting $f_{d-1}$, the two resulting integrals are
\[
2d \int_{\Delta^{d-2}} x_1^2 \cdots x_{d-2}^2 
    \prod_{1 \le i < j \le d-1} (x_i - x_j)^2 \, d\lambda^{d-2},
\]
and
\[
4d(d-2) \int_{\Delta^{d-2}} x_1 \cdots x_{d-3} 
    \prod_{i=1}^{d-1} x_i \prod_{1 \le i < j \le d-1} (x_i - x_j)^2 \, d\lambda^{d-2}.
\]
These integrals have the form of the integrals in Proposition~\ref{prop:main} in Appendix~\ref{appendix:selberg}, 
with the corresponding parameters. Thus the sum becomes
\begin{align*}
2d\, I_{\mathrm{Selberg}}(d-1,1,1,d-2,d-2) + 4d(d-2)\, I_{\mathrm{Selberg}}(d-1,2,1,0,d-3)
=  \\
=2d \Bigg(
    \frac{\Gamma(d+2)}{6 \, \Gamma(d^2-3)} \prod_{i=1}^{d} \Gamma(i)^2
  + 2(d-2) \frac{\Gamma(d+1)}{6 \, \Gamma(d^2-3)} \prod_{i=1}^{d} \Gamma(i)^2
  \Bigg)=  \\
= d(d-1) \frac{\Gamma(d+1)}{\Gamma(d^2-3)} \prod_{i=1}^{d} \Gamma(i)^2.
\end{align*}
Substituting this result into~\eqref{eq:V3} gives the desired formula for~\eqref{eq:main2}.

\end{proof}

\section{Complementarity polytope}\label{sec:complementarity}


Suppose the vectors $b_1,\ldots b_d$ form an orthonormal basis of $\mathbb C^d$, and let $P_k:=\ket{b_k}\bra{b_k}$ (i.e.\ the rank one orthogonal projection onto the one-dimensional subspace given by $b_k$) for $k=1,\ldots d$. Then the projections $P_k$ $(k=1,\ldots d)$, form a regular $d-1$-dimensional simplex in the state space $\mathbb S(\mathbb C^d)$, centered at $\frac{1}{d}I$ and having a (Hilbert-Schmidt) edge-length of $\sqrt{2}$.
For this reason, we shall consider an ``abstract'' $d-1$-dimensional simplex (i.e.\ one given in a Euclidean space, without reference to a concrete basis of $\mathbb C^d$). More concretely, let $\DeltaB^{d-1}$ denote the origin-centered $(d-1)$-dimensional simplex with edge length $\sqrt{2}$, given as the convex hull of the vectors
\[
v_1=\begin{bmatrix}
    -r_2 \\ -r_3 \\ -r_4 \\ \vdots \\ -r_{d}
\end{bmatrix}, 
v_2=\begin{bmatrix}
    R_2 \\ -r_3 \\ -r_4 \\ \vdots \\ -r_{d}
\end{bmatrix}, 
v_3=\begin{bmatrix}
    0 \\ R_3 \\ -r_4 \\ \vdots \\ -r_{d}
\end{bmatrix}, 
\cdots ,
v_{d-1}=\begin{bmatrix}
    0 \\ \vdots \\ 0 \\ R_{d-1} \\ -r_{d}
\end{bmatrix}, 
v_d=\begin{bmatrix}
    0 \\ \vdots \\ 0 \\ 0 \\ R_{d}
\end{bmatrix}
\]
where $r_d = 1/\sqrt{d(d-1)}$ and $R_d = \sqrt{(d-1)/d}$ denote the radii of the inscribed and circumscribed spheres, respectively. (We use the bold symbol $\DeltaB^{d-1}$ for this simplex to avoid confusion with the standard simplex $\Delta^{d-1}$.) 

Now let us define the complementarity polytope $\mathcal{P}_d \subset \mathbb{R}^{(d+1)d}$ as the convex hull of $d+1$ mutually orthogonal $(d-1)$-simplices, each centered at the origin. 
Let $\{e_1, e_2, \dots, e_{d+1}\} \subset \mathbb{R}^{d+1}$ denote the standard orthonormal basis of $\mathbb{R}^{d+1}$. 
Then the $i$-th simplex is given by
\[
\DeltaB^{d-1}_i := \operatorname{conv}\{\, e_i \otimes v_j \mid j \in [d] \,\} 
    \subset \mathbb{R}^{d-1}_i := e_i \otimes \mathbb{R}^{d-1},
\]
and the complementarity polytope itself is
\[
\mathcal{P}_d 
    := \operatorname{conv}\!\left( \bigcup_{i=1}^{d+1} \DeltaB^{d-1}_i \right) 
    = \operatorname{conv}\!\left( \bigcup_{i=1}^{d+1} \{\, e_i \otimes v_j \mid j \in [d] \,\} \right) 
    \subset \mathbb{R}^{d+1} \otimes \mathbb{R}^{d-1} .
\]
Note that $\mathbb{R}^{d+1} \otimes \mathbb{R}^{d-1} \cong \mathbb{R}^{d-1}_1 \oplus \cdots \oplus \mathbb{R}^{d-1}_{d+1}  \cong \mathbb{R}^{d^2 - 1}$.
We denote a face of the $i$-th simplex by specifying the set of omitted vertices $S_i \subseteq [d]$:
\[
\DeltaB^{d-1}_{i, S_i} := \operatorname{conv}\{\, e_i \otimes v_j \mid j \in [d] \setminus S_i \,\}.
\]
Using this notation, a face of $\mathcal{P}_d$ can be written as
\[
F(S_1, \dots, S_{d+1}) := 
    \operatorname{conv}\!\left( \bigcup_{i=1}^{d+1} \DeltaB^{d-1}_{i, S_i} \right).
\]
Due to the high degree of symmetry in the construction, all facets (i.e., faces of dimension $D-1 = d^2 - 2$) of $\mathcal{P}_d$ are of the same type, up to coordinate permutations and orthogonal transformations. 
In particular, each facet is isometric to
\[
F_{D-1} := F(\{1\}, \{1\}, \dots, \{1\}) 
    = \operatorname{conv}\!\left( \bigcup_{i=1}^{d+1} \DeltaB^{d-1}_{i, \{1\}} \right).
\]
Similarly, each $(D-2)$-dimensional face is isometric to
\[
\begin{aligned}
F_{D-2} 
    &:= F(\{1,2\}, \{1\}, \dots, \{1\}) \\
    &= \operatorname{conv}\!\left( 
        \DeltaB^{d-1}_{1, \{1,2\}} 
        \cup 
        \bigcup_{i=2}^{d+1} \DeltaB^{d-1}_{i, \{1\}} 
    \right)
    = \bigcap_{i=1}^2 F(\{i\}, \{1\}, \dots, \{1\}).
\end{aligned}
\]
However, there are two distinct types of $(D-3)$-dimensional faces:
\[
F_{D-3}^{(1)} := F(\{1,2,3\}, \{1\}, \dots, \{1\}) 
    = \bigcap_{i=1}^3 F(\{i\}, \{1\}, \dots, \{1\}),
\]
and
\[
F_{D-3}^{(2)} := F(\{1,2\}, \{1,2\}, \{1\}, \dots, \{1\}) 
    = \bigcap_{\substack{i=1 \\ j=1}}^2 F(\{i\}, \{j\}, \{1\}, \dots, \{1\}).
\]
In the derivations above, we have used the following identity, valid for all $k \in \mathbb{N}$:
\[
F\!\left( 
    \bigcup_{i=1}^k S_1^i, 
    \bigcup_{i=1}^k S_2^i, 
    \dots, 
    \bigcup_{i=1}^k S_{d+1}^i 
\right)
= 
\bigcap_{i=1}^k F(S_1^i, S_2^i, \dots, S_{d+1}^i).
\]

\begin{proof}[Proof of Theorem~\ref{th:P}]

First, we prove equation~\eqref{eq:P1}. 
Since $\mathcal{P}_d$ is a convex polytope, we can apply formula~\eqref{eq:VPoly}. 
(Note that, since we are calculating the \emph{unnormalized} intrinsic volume, it is not necessary to divide by $\operatorname{vol}_{2}(B^2)$.) 
As there is only one type of $(D-2)$-dimensional face—each is isomorphic to $F_{D-2}$—the sum simplifies to
\begin{equation}\label{eq:PV3}
\widetilde{V}_{D-2}(\mathcal{P}_d) 
    = f_{D-2}(\mathcal{P}_d)\,
      \operatorname{vol}_{D-2}(F_{D-2}) \cdot 
      \operatorname{vol}_{2}\!\big(N_{\mathcal{P}}(F_{D-2}) \cap B^2\big),
\end{equation}
where the first factor is the number of $(D-2)$-dimensional faces:
\begin{equation}\label{eq:PV3f}
f_{D-2}(\mathcal{P}_d) = (d+1)\binom{d}{2} d^d.
\end{equation}
To compute the second factor, the volume of the face $F_{D-2}$, we proceed indirectly by considering the cone formed by this face and the origin. 
The volume of this cone can be evaluated as 
(from now on, the notation ``$=\dots=$'' indicates omitted intermediate steps; detailed calculations can be found in Appendix~\ref{appendix:cpolytope}):
\[
\operatorname{vol}_{D-1}\!\big(\operatorname{conv}(\{0\} \cup F_{D-2})\big)
    = \dots = 
    \frac{\sqrt{2}}{(d^2 - 2)!} \frac{1}{d^{(d+1)/2}}.
\]
The corresponding height, i.e., the distance between the face and the origin, is
\[
\operatorname{dist}(0, F_{D-2})
    = \dots = 
    \frac{\sqrt{2}}{\sqrt{d(2d^2 - d - 2)}}.
\]
Using the standard formula for the volume of a cone in $\mathbb{R}^d$,
\[
\vol_d(\mathrm{cone}) = \frac{1}{d} \vol_{d-1}(\mathrm{\mathrm{base}}) \cdot \mathrm{height} ,
\]
we obtain
\begin{equation}\label{eq:PV3F}
\operatorname{vol}_{D-2}(F_{D-2}) 
    = (D-1)\,
      \frac{\operatorname{vol}_{D-1}\!\big(\operatorname{conv}(\{0\} \cup F_{D-2})\big)}
           {\operatorname{dist}(0, F_{D-2})}
    = \frac{1}{(d^2 - 3)!}\,
      \frac{\sqrt{2d^2 - d - 2}}{d^{d/2}}.
\end{equation}
The third factor is the two-dimensional volume of the normal cone:
\begin{equation}\label{eq:PV3N}
\operatorname{vol}_{2}\!\big(N_{\mathcal{P}}(F_{D-2}) \cap B^2\big)
    = \dots = 
    \frac{1}{2}\, \arccos\!\left(1 - \frac{d}{d^2 - 1}\right).
\end{equation}
Substituting the three factors~\eqref{eq:PV3f}, \eqref{eq:PV3F}, and~\eqref{eq:PV3N} into~\eqref{eq:PV3}, 
we obtain the first statement~\eqref{eq:P1} of Theorem~\ref{th:P}.

Now we prove equation~\eqref{eq:P2}. 
We follow the same approach as before; the only difference is that now there are two distinct types of $(D\!-\!3)$-dimensional faces (each isomorphic to either $F_{D-3}^{(1)}$ or $F_{D-3}^{(2)}$), 
so the sum in~\eqref{eq:VPoly} simplifies as
\begin{equation}\label{eq:PV4}
\widetilde{V}_{D-3}(\mathcal{P}_d)
  = \sum_{j=1}^2 
    f_{D-3}^{(j)}(\mathcal{P}_d)\,
    \operatorname{vol}_{D-3}\!\big(F_{D-3}^{(j)}\big)
    \cdot 
    \operatorname{vol}_{3}\!\big(N_{\mathcal{P}}(F_{D-3}^{(j)}) \cap B^3\big).
\end{equation}
The first factors are the numbers of $(D\!-\!3)$-dimensional faces of each kind:
\begin{align}\label{eq:PV4f}
\begin{split}
f_{D-3}^{(1)}(\mathcal{P}_d) &= (d+1)\binom{d}{3}d^d, \\
f_{D-3}^{(2)}(\mathcal{P}_d) &= \binom{d+1}{2} d^{d-1} \binom{d}{2}\binom{d}{2}.
\end{split}
\end{align}
The second factors can be computed by the same approach as before, 
using the volumes of the cones formed by each face and the origin:
\begin{align*}
\operatorname{vol}_{D-2}\!\big(\conv(\{0\}\cup F_{D-3}^{(1)})\big)
  &= \dots = \frac{\sqrt{3}}{(d^2-3)!} \frac{1}{\sqrt{d}^{\,d+1}},  \\
\operatorname{vol}_{D-2}\!\big(\conv(\{0\}\cup F_{D-3}^{(2)})\big)
  &= \dots = \frac{2}{(d^2-3)!} \frac{1}{\sqrt{d}^{\,d+1}}.
\end{align*}
The corresponding heights of these cones, i.e., the distances between the faces and the origin, are
\begin{align*}
\operatorname{dist}(0,F_{D-3}^{(1)}) &= \dots = \frac{1}{\sqrt{d\!\left(d^2 - \tfrac{2}{3}d - 1\right)}}, \\
\operatorname{dist}(0,F_{D-3}^{(2)}) &= \dots = \frac{1}{\sqrt{d(d^2 - d - 1)}}.
\end{align*}
Hence,
\begin{align}\label{eq:PV4F}
\begin{split}
\operatorname{vol}_{D-3}\!\big(F_{D-3}^{(1)}\big)
  &= (D-2)\,
     \frac{\operatorname{vol}_{D-2}\!\big(\conv(\{0\}\cup F_{D-3}^{(1)})\big)}
          {\operatorname{dist}(0,F_{D-3}^{(1)})}
     = \frac{\sqrt{3}}{(d^2-4)!}
       \frac{\sqrt{d\!\left(d^2-\tfrac{2}{3}d-1\right)}}{\sqrt{d}^{\,d+1}}, \\[4pt]
\operatorname{vol}_{D-3}\!\big(F_{D-3}^{(2)}\big)
  &= (D-2)\,
     \frac{\operatorname{vol}_{D-2}\!\big(\conv(\{0\}\cup F_{D-3}^{(2)})\big)}
          {\operatorname{dist}(0,F_{D-3}^{(2)})}
     = \frac{2}{(d^2-4)!}
       \frac{\sqrt{d(d^2-d-1)}}{\sqrt{d}^{\,d+1}}.
\end{split}
\end{align}
The third factors are the three-dimensional volumes of the normal cones:
\begin{align}\label{eq:PV4N}
\begin{split}
\operatorname{vol}_{3}\!\big(N_{\mathcal{P}}(F_{D-3}^{(1)}) \cap B^3\big)
  &= \dots = 
     \frac{4}{3}\,
     \arctan\!\left(
       \sqrt{
         \tan\!\left(\tfrac{3\alpha}{4}\right)
         \tan^3\!\left(\tfrac{\alpha}{4}\right)
       }
     \right), \\[4pt]
\operatorname{vol}_{3}\!\big(N_{\mathcal{P}}(F_{D-3}^{(2)}) \cap B^3\big)
  &= \dots = 
     \frac{8}{3}\,
     \arctan\!\left(
       \sqrt{
         \tan\!\left(\tfrac{\alpha}{2}+\tfrac{\beta}{4}\right)
         \tan\!\left(\tfrac{\alpha}{2}-\tfrac{\beta}{4}\right)
         \tan^2\!\left(\tfrac{\beta}{4}\right)
       }
     \right),
\end{split}
\end{align}
where 
\[
\alpha = \arccos\!\left(1 - \frac{d}{d^2 - 1}\right), 
\qquad 
\beta  = \arccos\!\left(1 - \frac{2d}{d^2 - 1}\right).
\]
Substituting the three factors~\eqref{eq:PV4f}, \eqref{eq:PV4F}, and~\eqref{eq:PV4N} into~\eqref{eq:PV4}, 
we obtain the second statement~\eqref{eq:P2} of Theorem~\ref{th:P}.

\end{proof}

\section{Example}\label{sec:example}

In the examples that follow, instead of a finite collections of vectors $v_1,\ldots v_n$ of a Euclidean space, we consider certain compact, continuously parameterized families of vectors that satisfy the same “trivial requirements’’ as discussed in the preliminaries. This slight abstraction is made for simplicity: in the continuous setting, the associated convex bodies -- such as intersections of Euclidean balls with coordinate orthants -- have volumes (and relevant intrinsic volumes) that can be computed more easily. Geometrically, this amounts to replacing the discrete polytopes considered before by convex bodies of simpler analytic description. By the continuity of intrinsic volumes under Hausdorff convergence, any such continuous configuration can be approximated arbitrarily well by a finite collection of vectors whose associated polytope would likewise be excluded by the same intrinsic-volume comparison. Hence, although we do not exhibit explicit finite configurations, the examples below genuinely represent the finite “would-be’’ configurations anticipated in the introduction.


For a (fixed) $d, D\equiv d^2-1$ and $k=0,1,2,3$, we consider the (four different) collection of vectors
\[
\Omega_{d,D-k} := \{\, v \in \mathbb{R}^{D-k} \mid \forall i: v_i \ge 0, \,\, \|v\| = R_d \,\},
\]
where \( R_d = \sqrt{(d-1)/d} \). In other words, \( \Omega_{d,D-k} \) is the intersection of the positive orthant with the sphere of radius \( R_d \) and origin at zero. 

It is immediate that the set \( \Omega_{D-k} \) satisfies the listed trivial requirements. Indeed, these vectors have the right length, the inner product of any two elements of \( \Omega_{d,D-k} \) is greater than or equal to $\frac{1}{d}$ -- it is actually always non-negative, as these vectors are all from one coordinate orthant -- and the span of these vectors is $(D-k)$-dimensional, which is indeed smaller than or equal to $d^2-1$. Therefore, by Lemma~3, the existence of a corresponding set of unit vectors \(v \mapsto \psi_v\)  in \( \mathbb{C}^d \) such that
\[
|\langle \psi_v, \psi_w \rangle|^2 = \langle v, w \rangle + \frac{1}{d}, \qquad (v,w \in \Omega_{D-k})
\]
is equivalent to asking whether the \emph{spherical cone}
\[
C_{d,D-k} := \mathrm{conv}(\{0\} \cup \Omega_{d,D-k}) = \{\,
v \in \mathbb{R}^{D-k} \mid \forall i: v_i \ge 0, \,\, \|v\| \leq R_d \,\} \]
can be (isometrically) inscribed into \( \mathcal{S}_d \) in such a way that the point \( 0 \in C_{d,D-k} \) is mapped to \( \frac{1}{d}I \in \mathcal{S}_d \). (Note that although Lemma~3 was formally stated for a finite collection of vectors, but in fact nowhere is used the finiteness in its proof.) We will show that in case $d=6$, such an inscription is impossible, comparing the intrinsic volumes \( V_{D-k}(C_{d,D-k}) \) with the (newly derived) intrinsic volumes \( V_{D-k}(\mathcal{S}_d) \).

Since \( C_{d,D-k} \) is a \((D-k)\)-dimensional convex body, its first nonzero intrinsic volume is its \((D-k)\)-dimensional volume, namely
\[
V_{D-k}(C_{d,D-k}) = \mathrm{vol}_{d,D-k}(C_{D-k}) = \frac{\chi_{D-k} R_d^{D-k}}{2^{D-k}}
= \frac{\left( \frac{\pi}{4}\frac{d-1}{d}\right)^{\frac{d^2-1-k}{2}}}{\Gamma\!\left(\frac{d^2-1-k}{2}+1\right)}
.
\]

Let us fix now $d$ to be $6$ and consider the four cases corresponding to the convex bodies \( C_{6,35}, C_{6,34}, C_{6,33} \), and \( C_{6,32} \). The inscription of the first example \( C_{6,35} \) can be ruled out by comparing volumes, since by straightforward substitution \( V_{35}(C_{6,35}) > V_{35}(\mathcal{S}_6) \). The second example \( C_{6,34} \) cannot be excluded by volume alone, because \( V_{35}(C_{6,34}) = 0\), but it can be ruled out using its second intrinsic volume (proportional to the surface area), as \( V_{34}(C_{6,34}) > V_{34}(\mathcal{S}_6) \). The third example \( C_{6,33} \) cannot be ruled out using either of the first two intrinsic volumes, since both vanish. However, by employing the newly derived quantity \( V_{D-2}(\mathcal{S}_d) \) the inscription into the state space can be excluded, as \( V_{33}(C_{6,33}) > V_{33}(\mathcal{S}_6) \). Finally, the inscription of the fourth example \( C_{6,32} \) can only be ruled out using its fourth intrinsic volume, since the first three vanish. Using the derived formula for \( V_{D-3}(\mathcal{S}_d) \), the nonexistence can again be established because \( V_{32}(C_{6,32}) > V_{32}(\mathcal{S}_6) \).
\smallskip

\textbf{Remarks:}
\begin{itemize}
\item We made the specific choice of $d=6$. In fact, things would work out exactly in the same way for any {\it larger} dimension. (Note that with the change of $d$ also the configuration whose existence is in question, changes.) However, e.g.\ for $d=5$ we have \( V_{21}(C_{5,21}) < V_{21}(\mathcal{S}_5) \), so in five dimensions the fourth example could not had been ruled out in the explained way.
\item 
Note that the dimensions of $3$ of the four examples are strictly smaller than that of the corresponding state space. By a continuity argument, it is easy to see that examples of the same dimension as the corresponding state space with similar relation of intrinsic volumes of the considered convex body and the quantum state space also exist, although they would be considerably more difficult to present and compute.
\end{itemize}

\appendix

\section{Selberg's integral}\label{appendix:selberg}

We summarize here some integral formulas of Selberg type following \cite{mehta1990random}, which are used in the evaluation of the simplex integrals in Section~\ref{sec:statespace}.

For any positive integer $n>1$, let $dx \equiv dx_1 \cdots dx_n$, and define
\[
\Phi(x) \equiv \Phi(x_1, \dots, x_n) 
  = \prod_{1 \le i < j \le n} (x_i - x_j)^{2\gamma} \prod_{k=1}^n x_k^{\alpha-1} e^{-x_k}.
\]
Then the following integral formulas are consequences of the Selberg integral of Laguerre type.  Then, for integers $n>1$ and $1 \le m \le n$, and real parameters $\alpha > 0$ and 
$\gamma > -\min\{1/n, \, \alpha/(n-1)\}$, the following integral formulas hold:
\begin{align}\label{eq:selberg}
\int_{[0, \infty)^n} \Phi(x)\, dx 
  &= \prod_{j=0}^{n-1} \frac{\Gamma(1 + (1+j)\gamma)\, \Gamma(\alpha + j \gamma)}{\Gamma(1+\gamma)},  \nonumber \\ 
\int_{[0, \infty)^n} x_1^2 \cdots x_k^2 \, x_{k+1} \cdots x_m \, \Phi(x)\, dx 
  &= \prod_{j=1}^{k} (\alpha + 1 + \gamma (2n - m - j)) \cdot  \\
     & \ \  \cdot \prod_{j=1}^{m} (\alpha + \gamma (n - j)) 
     \int_{[0, \infty)^n} \Phi(x)\, dx.  \nonumber
\end{align}

\begin{lemma}\label{lemma:integral_simplex_exp}
Let $h$ be a homogeneous polynomial of degree $k$ in $n$ variables, and let 
$\Delta^{n-1} \subset \mathbb{R}^n$ denote the standard $(n-1)$-simplex. Define
\begin{align*}
I_\Delta &= \int_{\Delta^{n-1}} h(x) \, d\lambda^{n-1}(x), \\
I_e &= \int_{[0, \infty)^n} e^{-\sum_{i=1}^n x_i} \, h(x) \, d\lambda^n(x),
\end{align*}
where $d\lambda^{n-1}(x)$ denotes the Lebesgue measure induced on the simplex. Then
\[
I_\Delta = \frac{\sqrt{n}}{\Gamma(n + k)} \, I_e.
\]
\end{lemma}

\begin{proof}
Consider the Lipschitz function
\[
u(x) = \sum_{i=1}^n x_i
\]
on $\mathbb{R}^n$, which partitions $[0, \infty)^n$ into level sets
\[
u^{-1}(t) \cap [0, \infty)^n = t \, \Delta^{n-1} = \{ x \in [0, \infty)^n : \sum_{i=1}^n x_i = t \}.
\]
Thus, by the \emph{coarea} formula (\cite{morganGMT, federerGMT}), the integral of the $L^1$ function $e^{-u(x)} h(x)$ over $[0, \infty)^n$ can be written as
\begin{align*}
\int_{[0, \infty)^n} e^{-u(x)} h(x) \left| \nabla u(x) \right| \, d\lambda^n(x) 
&= \int_0^\infty \left( \int_{t\Delta^{n-1}} e^{-t} h(x) \, d\lambda^{n-1}(x) \right) dt \\
&= \int_0^\infty \left( t^{k+n-1} e^{-t} \int_{\Delta^{n-1}} h(x) \, d\lambda^{n-1}(x) \right) dt \\
&= I_\Delta \int_0^\infty t^{k+n-1} e^{-t} \, dt = I_\Delta \Gamma(n+k).
\end{align*}
In the second equality, we use the change of variables $x = t \tilde{x}$, $\tilde{x} \in \Delta^{n-1}$, and the homogeneity of $h$, i.e., $h(t \tilde{x}) = t^k h(\tilde{x})$, together with the scaling of the volume element $d\lambda^{n-1}(x) = t^{n-1} d\lambda^{n-1}(\tilde{x})$. Since $|\nabla u(x)| = \sqrt{n}$, the left-hand side equals $\sqrt{n} \, I_e$, which proves the lemma.
\end{proof}

\begin{proposition}\label{prop:main}
Let $\Delta^{n-1} \subset \mathbb{R}^n$ denote the standard $(n-1)$-simplex, and let $d\lambda^{n-1}(x)$ denote the induced Lebesgue measure. Then
\begin{align*}
&I_{\mathrm{Selberg}}(n,\alpha,\gamma,k,m) :=\!\int_{\Delta^{n-1}} \! x_1^2,\ldots x_k^2 x_{k+1} \ldots x_m \prod_{i=1}^n x_i^{\alpha-1} \prod_{i<j}^n (x_i-x_j)^{2\gamma} d\lambda^{n-1}(x)= \\
&=\sqrt{n} \frac{\prod_{j=1}^k (\alpha+1+\gamma(2n-m-j)) \prod_{j=1}^m (\alpha+\gamma(n-j))}{\Gamma (m+k+n(\alpha+\gamma(n-1))}  \prod_{j=0}^{n-1} \frac{\Gamma (1+\gamma (j+1) \Gamma (\alpha + \gamma j))}{\Gamma (1+\gamma)}.
\end{align*}
\end{proposition}

\begin{proof}
The integrand is homogeneous of degree $m + k + n(\alpha - 1 + \gamma(n-1))$. By Lemma~\ref{lemma:integral_simplex_exp}, the integral reduces to a constant multiple of the corresponding Selberg-type integral. Applying the known evaluation of Selberg integrals~\eqref{eq:selberg} then yields the stated formula.
\end{proof}

\section{Detailed calculations for Theorem 2} \label{appendix:cpolytope}
We now detail the determinant-based calculations used to obtain the cone volumes in Theorem~\ref{th:P}. 
Specifically, we compute the volume of the cone formed by a particular face of $\mathcal{P}_d$ and the origin.

\medskip

The vertices of the facet $\DeltaB^{d-1}_{\{1\}}$, expressed as column vectors, form the following 
$(d\!-\!1)\times(d\!-\!1)$ matrix:
\[
M_1 :=
\begin{pmatrix}
v_2\;\big|\;v_3\;\big|\;\cdots\;\big|\;v_d
\end{pmatrix}
=
\begin{pmatrix}
  R_2      & 0        & 0      & \cdots & 0      \\[6pt]
  -r_3     & R_3      & 0      & \cdots & 0      \\[6pt]
  -r_4     & -r_4     & R_4    & \ddots & \vdots \\[3pt]
  \vdots   & \vdots   & \ddots & \ddots & 0      \\[3pt]
  -r_{d}   & -r_{d}   & \cdots & -r_{d} & R_{d}
\end{pmatrix}.
\]
Let $M_k = M_{k:d,\,k:d}$ denote the $(d\!-\!k)\times(d\!-\!k)$ submatrix of $M_1$ obtained by deleting the first $k$ rows and columns. 
Its determinant satisfies
\[
\det M_k = \prod_{j=k+1}^{d} R_j = \sqrt{\frac{k}{d}}.
\]

\medskip

The face $F_{D-2}$ is spanned by the column vectors of the block-diagonal matrix 
$\operatorname{diag}(M_2, M_1, \dots, M_1)$. 
Hence, the volume of the cone over this face is
\[
\operatorname{vol}_{D-1}\!\big(\conv(\{0\}\cup F_{D-2})\big)
= \frac{1}{(D-1)!}\, (\det M_2)(\det M_1)^d
= \frac{\sqrt{2}}{(d^2-2)!} \frac{1}{\sqrt{d}^{\,d+1}}.
\]
Similarly, the faces $F_{D-3}^{(1)}$ and $F_{D-3}^{(2)}$ are spanned by the column vectors of 
$\operatorname{diag}(M_3, M_1, \dots, M_1)$ 
and 
$\operatorname{diag}(M_2, M_2, M_1, \dots, M_1)$, 
respectively. 
Their corresponding cone volumes are:
\[
\operatorname{vol}_{D-2}\!\big(\conv(\{0\}\cup F_{D-3}^{(1)})\big)
= \frac{1}{(D-2)!}\, (\det M_3)(\det M_1)^d
= \frac{\sqrt{3}}{(d^2-3)!} \frac{1}{\sqrt{d}^{\,d+1}},
\]
\[
\operatorname{vol}_{D-2}\!\big(\conv(\{0\}\cup F_{D-3}^{(2)})\big)
= \frac{1}{(D-2)!}\, (\det M_2)^2(\det M_1)^{d-1}
= \frac{2}{(d^2-3)!} \frac{1}{\sqrt{d}^{\,d+1}}.
\]
Before calculating the distances between specific faces of $\mathcal{P}_d$ and the origin, 
we begin with two auxiliary lemmas.

\begin{lemma}\label{lemma:closest}
Let $V_1, \dots, V_k$ be mutually orthogonal subspaces of $\mathbb{R}^d$, 
and let $C_i \subset V_i$ be closed convex sets such that $0 \notin C_i$ for all $i \in [k]$. 
Then the point in 
\[
C := \conv(C_1 \cup \cdots \cup C_k)
\]
closest to the origin is given by
\[
x = \sum_{i=1}^k \lambda_i x_i,
\qquad 
\lambda_i = \frac{1 / \|x_i\|^2}{\sum_{j=1}^k 1 / \|x_j\|^2},
\]
where $x_i \in C_i$ denotes the point in $C_i$ closest to the origin. 
Moreover,
\[
\frac{1}{\|x\|^2} = \sum_{j=1}^k \frac{1}{\|x_j\|^2}.
\]
\end{lemma}

\begin{proof}
Since the subspaces $V_i$ are mutually orthogonal, any point 
$y \in C$ that is a convex combination of points $y_i \in C_i \subset V_i$ satisfies
\[
\norm{\sum \lambda_i y_i}^2=\sum_i \lambda_i^2 \norm{y_i}^2.
\]
Therefore, the point in $C$ closest to the origin must be a convex combination 
of the individual points $x_i \in C_i$ that are closest to the origin, 
that is, $x = \sum_i \lambda_i x_i$. 
The problem thus reduces to finding the point in 
$\conv\{x_1, \dots, x_k\} \subset \mathbb{R}^k$ closest to the origin. Since all $x_i$ lie in mutually orthogonal directions, this convex hull 
lies in the affine hyperplane
\[
H := \left\{\, y \in \mathbb{R}^k : \langle v, y \rangle = 1 \,\right\},
\qquad 
v := \sum_{j=1}^k \frac{x_j}{\|x_j\|^2},
\]
whose normal vector is $v$. 
The closest point on this hyperplane to the origin 
is the orthogonal projection of the origin onto $H$, 
which is given by $x = v / \|v\|^2$. 
This yields the stated expressions for $x$ and $\|x\|$, completing the proof.
\end{proof}
We next recall a simple geometric fact about the standard simplex, 
which will be used to compute the distances between faces of $\mathcal{P}_d$ and the origin.

\begin{lemma}\label{lemma:facedist}
Let $\Delta^{d-1}$ be the standard $(d-1)$-dimensional simplex in $\mathbb{R}^d$. 
Then the Euclidean distance between the centroid of $\Delta^{d-1}$ 
and the centroid of any $(d-1-k)$-dimensional face is
\[
r(d, k) = \sqrt{\frac{k}{d(d - k)}}.
\]
In particular, for $k = 1$ and $k = d - 1$, this recovers the radii 
of the inscribed and circumscribed spheres of the simplex, respectively.
\end{lemma}
\begin{proof}
Let $\Delta^{d-1}$ denote the standard simplex in $\mathbb{R}^d$ with vertices given by the standard basis vectors $e_1, \dots, e_d$. The centroid of the full simplex is
\[
c=\frac{1}{d}\sum_{i=1}^d e_i,
\]
and, without loss of generality, the centroid of the $(d-1-k)$-face determined by $e_1,\dots,e_{d-k}$ is
\[
p=\frac{1}{d-k}\sum_{i=1}^{d-k} e_i.
\]
A straightforward computation yields
\[
\|c-p\|^2=\frac{k}{d(d-k)},
\]
which proves the claim.
\end{proof}
Now consider the face
\[
F_{D-2}= \operatorname{conv}\!\Big( \DeltaB^{d-1}_{1, \{1,2\}} \cup \bigcup_{i=2}^{d+1} \DeltaB^{d-1}_{i, \{1\}} \Big),
\]
which is the convex hull of orthogonal faces of simplices. Applying 
Lemmas~\ref{lemma:closest} and~\ref{lemma:facedist} to compute the distance from the origin yields
\[
\operatorname{dist}(0,F_{D-2})^2
= \frac{1}{d\cdot \frac{1}{r_d^2} + \frac{1}{r(d,2)^2}}
= \frac{2}{d(2d^2-d-2)}.
\]
Similarly,
\begin{align*}
\dist(0,F_{D-3}^{(1)})^2&=\frac{1}{d \frac{1}{r_d^2} + \frac{1}{r(d,3)^2}}  = \frac{1}{d(d^2-\frac{2}{3}d-1)}, \\
\dist(0,F_{D-3}^{(2)})^2&=\frac{1}{(d-1) \frac{1}{r_d^2} + 2\frac{1}{r(d,2)^2}} =\frac{1}{d(d^2-d-1)}.
\end{align*}
We now calculate the volumes of the normal cones corresponding to the relevant faces. Let $u_i$ denote the unit outward normal to the facet $\DeltaB_{\{i\}}^{d-1}$ of the simplex $\DeltaB^{d-1}$. Then
\[
\langle u_i,u_i\rangle=1,\qquad \langle u_i,u_j\rangle=\frac{1}{1-d}\quad (i\ne j).
\]
Let $w(j_1,\dots,j_{d+1})$ denote the outer unit normal vector of the facet $F(\{j_1\},\dots,\{j_{d+1}\})$ of $\mathcal{P}_d$, where $j_i\in [d]$. Then the relevant normals (up to permutation) are
\begin{align*} 
w(1,1,\dots,1)&:=\frac{1}{\sqrt{d+1}} \sum_{i=1}^{d+1} e_i \otimes u_1,\\
w(2,1,\dots,1)&:=\frac{1}{\sqrt{d+1}}\Big(e_1\otimes u_2 + \sum_{i=2}^{d+1} e_i\otimes u_1\Big),\\
w(3,1,\dots,1)&:=\frac{1}{\sqrt{d+1}}\Big(e_1\otimes u_3 + \sum_{i=2}^{d+1} e_i\otimes u_1\Big),\\
w(1,2,1,\dots,1)&:=\frac{1}{\sqrt{d+1}}\Big(e_1\otimes u_1 + e_2\otimes u_2 + \sum_{i=3}^{d+1} e_i\otimes u_1\Big),\\
w(2,2,1,\dots,1)&:=\frac{1}{\sqrt{d+1}}\Big(\sum_{i=1}^{2} e_i\otimes u_2 + \sum_{i=3}^{d+1} e_i\otimes u_1\Big).
\end{align*}
A straightforward inner-product computation yields
\begin{align*} 
\braket{w(1,1,\dots,1),w(2,1,\dots,1)}&=\frac{1}{d+1}(\braket{u_1,u_2}+d\braket{u_1,u_1})=1-\frac{d}{d^2-1} ,\\
\braket{w(1,1,\dots,1),w(2,2,\dots,1)}&=\frac{1}{d+1}(2\braket{u_1,u_2}+(d-1)\braket{u_1,u_1})=1-\frac{2d}{d^2-1}.
\end{align*}
Denote the corresponding angles by
\[
\alpha=\arccos\!\Big(1 - \frac{d}{d^2 - 1}\Big),\qquad 
\beta=\arccos\!\Big(1 - \frac{2d}{d^2 - 1}\Big).
\]

The normal cone $N_{\mathcal{P}}(F_{D-2})\cap B^2$ is generated by the first two normal vectors above and hence spans an angle $\alpha$. Its two-dimensional spherical volume (i.e. the spherical arc length on the unit circle in the normal plane) is
\[
\operatorname{vol}_{2}\!\big(N_{\mathcal{P}}(F_{D-2})\cap B^2\big) = \tfrac{1}{2}\,\alpha.
\]

The normal cone $N_{\mathcal{P}}(F_{D-3}^{(1)})\cap B^3$ is generated by the first three normal vectors and corresponds to a (spherical) equilateral triangle on the unit sphere with side lengths all equal to $\alpha$. Using L'Huilier's formula for the area \(A\) of a spherical triangle with side-lengths \(a,b,c\) (on the unit sphere),
\[
\tan\!\left(\frac{A}{4}\right)
= \sqrt{ \tan\!\left(\frac{s}{2}\right) \tan\!\left(\frac{s-a}{2}\right) \tan\!\left(\frac{s-b}{2}\right) \tan\!\left(\frac{s-c}{2}\right) },
\qquad s=\frac{a+b+c}{2},
\]
and taking \(a=b=c=\alpha\), we obtain
\[
\operatorname{vol}_{3}\!\big(N_{\mathcal{P}}(F_{D-3}^{(1)})\cap B^3\big)
= \frac{4}{3}\,\arctan\!\left(\sqrt{\tan\!\left(\tfrac{3\alpha}{4}\right)\tan^3\!\left(\tfrac{\alpha}{4}\right)}\right).
\]
Finally, the normal cone $N_{\mathcal{P}}(F_{D-3}^{(2)})\cap B^3$ is generated by normals of the mixed type and forms a spherical quadrilateral which can be decomposed into two spherical triangles with side lengths \(\alpha\) and \(\beta\). Applying L'Huilier's formula to each triangle and combining the results yields
\[
\operatorname{vol}_{3}\!\big(N_{\mathcal{P}}(F_{D-3}^{(2)})\cap B^3\big)
= \frac{8}{3}\,\arctan\!\left(\sqrt{\tan\!\Big(\tfrac{\alpha}{2}+\tfrac{\beta}{4}\Big)\,\tan\!\Big(\tfrac{\alpha}{2}-\tfrac{\beta}{4}\Big)\,\tan^2\!\Big(\tfrac{\beta}{4}\Big)}\right).
\]

\section*{Acknowledgment}

We thank László Szilágyi for his handy and efficient implementation of the program~\cite{our_code}, which enabled us to verify our formulas to several decimal places across multiple dimensions and may prove useful for further investigations. We are also grateful to Márton Naszódi for introducing us to the concept of intrinsic volumes. Finally, we thank the organizers of the Hadamard 2025 conference in Seville, where our results were first presented.

\bibliographystyle{amsalpha} 
\bibliography{references}
\end{document}